\documentclass{article}
\usepackage[utf8]{inputenc}
\usepackage{amsmath} 
\usepackage{amsfonts}
\usepackage{amssymb}
\usepackage{amsthm}

\newtheorem*{theorem}{Theorem}
\newtheorem{lemma}{Lemma}
\title{Matching on the line admits\\
no $o(\sqrt{\log n})$-competitive algorithm}
\author{Enoch Peserico\thanks{Universit\`a degli Studi di Padova, Italy. Corresponding email: \texttt{scquizza@math.unipd.it}.} \and Michele Scquizzato\protect\footnotemark[1]  
}
\date{}

\begin{document}
\maketitle
\begin{abstract}
\noindent We present a simple proof that the competitive ratio of any randomized online matching algorithm for the line is at least $\sqrt{\log_2(n\!+\!1)}/12$
for all $n=2^i\!-\!1: i\in\mathbb{N}$.
\end{abstract}

\section{Online matching, on the line}
In \emph{online metric matching}~\cite{KalyanasundaramP93,KhullerMV94} $n$ points of a metric space are designated as \emph{servers}. One by one $n$ \emph{requests} arrive at arbitrary points of the space; upon arrival each must be matched to a yet unmatched server, at a cost equal to their distance. Matchings should minimize the ratio between the total cost and the \emph{offline} cost attainable if all requests were known beforehand. A matching algorithm is \emph{$c(n)$-competitive} if it keeps this ratio no higher than $c(n)$ for all possible placements of servers and requests. 

It is widely acknowledged~\cite{AntoniadisBNPS19,KoutsoupiasN03,vanStee16a} that the line is the most interesting metric space for the problem. Matching on the line models many scenarios, like a shop that must rent to customers skis of approximately their height, where a stream of requests must be serviced with minimally mismatched items from a known store. 
Despite matching being specifically studied  on the line
 since at least 1996~\cite{KalyanasundaramP98}, no tight competitiveness bounds are known.

As for upper bounds, the line is a doubling space and thus admits an $O(\log n)$-competitive randomized algorithm~\cite{GuptaL12}; 
a sequence of recent developments~\cite{AntoniadisBNPS19,NayyarR17,Raghvendra18} yielded the same ratio 
without randomization. Better bounds have been obtained only by algorithms with additional power, 
such as that to re-assign past requests~\cite{GuptaKS20,MegowN20} or predict future ones~\cite{AntoniadisCE0S20}.

As for lower bounds, the competitive ratio is at least $4.591$ for randomized algorithms and $9$ for deterministic ones since the \emph{cow-path} problem is a special case of matching on the line~\cite{KalyanasundaramP98}. These bounds were conjectured  tight~\cite{KalyanasundaramP98} until a complex adversarial strategy yielded a lower bound of $9.001$ for deterministic algorithms~\cite{FuchsHK05}. 
Beyond some $\Omega(\log n)$ lower bounds for restricted classes of algorithms~\cite{AntoniadisFT18, KoutsoupiasN03,NayyarR17}, no further progress has been made on the lower-bound side in the last two decades.

\section{An $\Omega(\sqrt{\log n})$-competitiveness bound}
We prove a simple $\Omega(\sqrt{\log n})$ lower bound on the competitive ratio of randomized online matching algorithms for the line.

For any $n=2^i-1$ with $i\in\mathbb N$ consider the $[0,n\!+\!1]$ interval; for each positive integer $j\leq n$ place a server at point $j$, and place $n$ requests over $\log_2(n\!+\!1)$ rounds as follows. For $r=1,\dots,\log_2(n\!+\!1)$, partition the interval into $(n\!+\!1)/2^r$ subintervals of length $2^r$ and choose for each request a point uniformly and independently at random in each subinterval; call that point the request's \emph{origin} and place the actual request on the closest integer multiple of $2^{-n}$, breaking ties arbitrarily (this ``discretization'' prevents technical difficulties, see Remark 1).

We prove in Lemma~\ref{lem:ub} that the expected distance between the $\ell^{th}$ leftmost server and the $\ell^{th}$ leftmost origin is $O(\sqrt{\log n})$, so servers and requests can be matched with an expected offline cost $O(n \sqrt{\log n})$. Conversely, we prove in Lemma~\ref{lem:lb} that this request distribution forces \emph{any} online matching algorithm ALG to incur an expected $\Omega(n)$ cost 
in any given round, for a total cost $\Omega(n\log n)$. This is true even if ALG can examine all requests in a round before servicing any. 
The two results can be combined to prove that for \emph{some} request sequence ALG incurs $\Omega(\sqrt{\log n})$ times the offline cost.

\begin{lemma}
\label{lem:ub}
The expected distance between the $\ell^{th}$ leftmost origin and the $\ell^{th}$ leftmost server is at most $\sqrt{\log_2 (n\!+\!1)}+3$.
\end{lemma}
\begin{proof}
Let $S_\ell$ be $\ell^{th}$ leftmost server and
$g_\ell$ be the number of origins to its left. Note that if $g_\ell$ equals respectively $\ell$ or $\ell-1$, the $\ell^{th}$ origin is the first immediately to the left, or to the right of $S_\ell$; and since the first round placed one origin in every subinterval of size $2$, such an origin is within distance $3$ of $S_\ell$. By the same token, denoting by $\delta_\ell$ the quantity $|g_\ell- (\ell-\frac{\ell}{n\!+\!1})|$, for any $\delta$ such that $\delta_\ell\leq \delta$ the $\ell^{th}$ leftmost origin is within distance $2\delta+3$ of $S_\ell$. Note that $\delta_\ell$ is the absolute deviation from the mean of $r_\ell$, since $r_\ell$ is the sum of $n$ independent indicator random variables each denoting whether a given origin was placed to the left of $S_\ell$, with total expectation $\frac{n}{n\!+\!1} \ell = \ell -\frac{\ell}{n\!+\!1}$ (by construction, the expected density of origins is constant throughout the main interval). At most one such variable in a given round has variance greater than $0$, albeit obviously at most $1/4$: that corresponding to the origin placed in a subinterval holding $S_\ell$ strictly in its interior. Adding the individual variances we obtain the variance of $r_\ell$, i.e.\ the expectation of $\delta_\ell^2$, is at most $\log_2 (n\!+\!1)/4$; and since by Jensen's inequality $E[\delta_\ell]\leq E[\delta_\ell^2]^{\frac{1}{2}}$, the expected distance between $S_\ell$ and the $\ell^{th}$ leftmost origin is at most $\sqrt{\log_2 (n\!+\!1)}+3$.
\end{proof}

\begin{lemma}
\label{lem:lb}
Any randomized online matching algorithm incurs an expected cost greater than $(n\!+\!1)/12$
in any given round, even knowing all requests in each round before servicing any. 
\end{lemma}
\begin{proof}
Consider an origin placed uniformly at random in a subinterval of size $2^r$ during the $r^{th}$ round. Assume $m$ unmatched servers in the interior points of that subinterval divide it into $m+1$ segments of (integer) length $d_0,\dots,d_m$. Then the probability the corresponding \emph{request} falls within a segment of length $d$ is $d/2^r$, in which case the expected distance of the \emph{request} from the segment's closer endpoint is $d/4$. Adding over all the $s_r$ segments
in 
all the round's subintervals, and noting that $s_r$ does not exceed the  number of subintervals (i.e. $(n\!+\!1)/2^r$) plus the total number of unmatched servers (i.e. $(n\!+\!1)/2^{r-1}-1$), the expected cost to service 
all  requests in the round is at least:

\begin{equation*}
\sum_{h=1}^{s_r} \frac{d_h}{4} \cdot \frac{d_h}{2^r} \geq \frac{1}{4\cdot 2^r}s_r\left(\frac{n\!+\!1}{s_r}\right)^2
> \frac{(n\!+\!1)^2}{4\cdot 2^r} \cdot \frac{2^r}{3(n\!+\!1)} 
= \frac{n\!+\!1}{12}.
\end{equation*}

\end{proof}

We can then easily prove the following:

\begin{theorem}
The competitive ratio of any randomized online matching algorithm for the line is at least $\sqrt{\log_2(n\!+\!1)}/12$ for all $n=2^i-1: i\in\mathbb{N}$. 
\end{theorem}
\begin{proof}
Let $C_{A}(\sigma)$ be the expected cost incurred by a randomized online matching algorithm ALG on a request sequence $\sigma$, and $C_{O}(\sigma)$ the offline cost; and let $p_\sigma$ be the probability of generating $\sigma$ through the origin-request process described earlier. 
Since $\forall a_i, b_i>0$ we have that $(\sum_i a_i)/(\sum_i b_i)$ is a convex linear combination of the individual ratios $a_i/b_i$, then if $\sqrt{\log_2(n\!+\!1)}>12$:
\begin{equation*}
\max_{\sigma:p_\sigma\neq 0} \frac{C_{A} (\sigma)}{C_{O} (\sigma)}
\geq \frac{\sum\limits_{\sigma:p_\sigma\neq 0} C_{A} (\sigma) p_\sigma}
{\sum\limits_{\sigma:p_\sigma\neq 0} C_{O} (\sigma) p_\sigma}
>\frac{(n\!+\!1)\log_2(n\!+\!1)/12}
{n\sqrt{\log_2 (n\!+\!1)}+3+n2^{-n}}
\geq\frac{\sqrt{\log_2 (n\!+\!1)}}{12}.
\end{equation*}
\end{proof}

\textbf{Remark 1:} Without discretization the term $\sum_{\sigma:p_\sigma\neq 0} C_{A} (\sigma) p_\sigma$ in the proof of the theorem  would become an integral potentially not well-defined for some pathologic algorithms,
e.g.\ those servicing requests for rational points in an interval with one server and for irrational points with another.

\textbf{Remark 2:} For any constant $\epsilon>0$ ALG incurs an $\Omega(n \log n)$ cost even over the last $\epsilon \log_2 n$ rounds; thus, the bound holds even if ALG has advance knowledge of all but the last $n^{\epsilon}$ requests.

\textbf{Remark 3:} Consider a space with a non-zero metric. For any set $X$ of $n$ points, let $d_X$ be its diameter and $c_X$ the length of its shortest circuit; and let $\mu = \sup_X c_X/d_X$ (e.g.\ $\mu=2$ on the line). No online matching algorithm can be $o(\mu)$-competitive, and there is an algorithm $O(\mu\log^2 n)$-competitive in \emph{all} spaces~\cite{NayyarR17}. It has been asked whether some similarly universal algorithm can have an $O(\mu)$-competitive ratio~\cite{NayyarR17}: our result provides a negative answer.

\section*{Acknowledgements}
\noindent  Michele Scquizzato is supported, in part, by Univ.\ Padova grant BIRD197859/19. Both authors thank Kirk Pruhs for his constructive criticism and his insightful observations, including that leading to Remark 3.


\begin{thebibliography}{10}

\bibitem{AntoniadisBNPS19}
A.~Antoniadis, N.~Barcelo, M.~Nugent, K.~Pruhs, and M.~Scquizzato.
\newblock A $o(n)$-competitive deterministic algorithm for online matching on a line.
\newblock {\em Algorithmica}, 81(7):2917--2933, 2019.

\bibitem{AntoniadisCE0S20}
A.~Antoniadis, C.~Coester, M.~Eli{\'{a}}s, A.~Polak, and B.~Simon.
\newblock Online metric algorithms with untrusted predictions.
\newblock In {\em Proceedings of the 37th ICML}, pages 345--355, 2020.

\bibitem{AntoniadisFT18}
A.~Antoniadis, C.~Fischer, and A.~T{\"{o}}nnis.
\newblock A collection of lower bounds for online matching on the line.
\newblock In {\em Proceedings of the 13th LATIN}, pages 52--65, 2018.

\bibitem{FuchsHK05}
B.~Fuchs, W.~Hochst{\"a}ttler, and W.~Kern.
\newblock Online matching on a line.
\newblock {\em Theor. Comput. Sci.}, 332(1-3):251--264, 2005.

\bibitem{GuptaL12}
A.~Gupta and K.~Lewi.
\newblock The online metric matching problem for doubling metrics.
\newblock In {\em Proceedings of the 39th ICALP}, pages 424--435, 2012.

\bibitem{GuptaKS20}
V.~Gupta, R.~Krishnaswamy, and S.~Sandeep.
\newblock Permutation strikes back: The power of recourse in online metric matching.
\newblock In {\em Proceedings of the 23rd APPROX}, pages 40:1--40:20, 2020.

\bibitem{KalyanasundaramP93}
B.~Kalyanasundaram and K.~Pruhs.
\newblock Online weighted matching.
\newblock {\em J. Algorithms}, 14(3):478--488, 1993.

\bibitem{KalyanasundaramP98}
B.~Kalyanasundaram and K.~Pruhs.
\newblock Online network optimization problems.
\newblock In {\em Online Algorithms: The State of the Art}, pages 268--280.
  Springer-Verlag, 1998.
\newblock From the Dagstuhl Seminar on Online Algorithms, 1996.

\bibitem{KhullerMV94}
S.~Khuller, S.~G. Mitchell, and V.~V. Vazirani.
\newblock On-line algorithms for weighted bipartite matching and stable marriages.
\newblock {\em Theor. Comput. Sci.}, 127(2):255--267, 1994.

\bibitem{KoutsoupiasN03}
E.~Koutsoupias and A.~Nanavati.
\newblock The online matching problem on a line.
\newblock In {\em Proceedings of the 1st WAOA}, pages 179--191, 2003.

\bibitem{MegowN20}
N.~Megow and L.~N{\"{o}}lke.
\newblock Online minimum cost matching with recourse on the line.
\newblock In {\em Proceedings of the 23rd APPROX}, pages 37:1--37:16, 2020.

\bibitem{NayyarR17}
K.~Nayyar and S.~Raghvendra.
\newblock An input sensitive online algorithm for the metric bipartite matching problem.
\newblock In {\em Proceedings of the 58th IEEE FOCS}, pages 505--515, 2017.

\bibitem{Raghvendra18}
S.~Raghvendra.
\newblock Optimal analysis of an online algorithm for the bipartite matching problem on a line.
\newblock In {\em Proceedings of the 34th SoCG}, pages 67:1--67:14, 2018.

\bibitem{vanStee16a}
R.~van Stee.
\newblock {SIGACT} news online algorithms column 27: Online matching on the line, part 1.
\newblock {\em {SIGACT} News}, 47(1):99--110, 2016.

\end{thebibliography}

\end{document}